\documentclass[11pt]{amsart}
\usepackage{fullpage}
\usepackage{mydefs}
\usepackage{amssymb}
\usepackage{graphicx}

\usepackage{amsmath}

\usepackage{graphicx}
\usepackage{amssymb}
\usepackage{multicol}
\usepackage{epstopdf}
\usepackage{mydefs}
\usepackage{algorithm}
\usepackage{algorithmic}

\usepackage{url}

\def\calL{{\mathcal L}}   
\def\calP{{\mathcal P}}
\def\calQ{{\mathcal Q}}      

\title{Wireless Capacity and Admission Control in Cognitive Radio}

\author[M. M. Halld\'orsson]{Magn\'us M. Halld\'orsson}
\address[M. M. Halld\'orsson]{School of Computer Science\\
 Reykjavik University\\
 101 Reykjavik, Iceland\\}
\email{mmh@ru.is}

\author[P. Mitra]{Pradipta Mitra}
\address[P. Mitra]{School of Computer Science\\
Reykjavik University\\
Reykjavik 101, Iceland}
\email{ppmitra@gmail.com}

\begin{document}

\begin{abstract}
We give algorithms with constant-factor performance guarantees for
several capacity and throughput problems in the SINR model.  The
algorithms are all based on a novel LP formulation for capacity
problems.  First, we give a new constant-factor approximation
algorithm for selecting the maximum subset of links that can be
scheduled simultaneously, under any non-decreasing and sublinear power assignment.
For the case of uniform power, we extend this to the case of variable QoS requirements and link-dependent
noise terms. 
Second, we approximate a problem related to cognitive radio: find a maximum set of links that can be simultaneously scheduled without affecting a given set of previously assigned links.
Finally, we obtain constant-factor approximation of weighted capacity under linear power assignment.

\end{abstract}

\maketitle 
\section{Introduction}
How much communication can be active simultaneously in a given wireless network?
This is a topic of major research effort. We address this question in
a more generalized setting than previously considered, and give
efficient algorithms that achieve good performance guarantees based on
a novel mathematical programming formulation. 

In the \emph{capacity} problem in wireless networks, we are given a
set of communication links in a metric space, each consisting of
a sender-receiver pair, and we seek to find the largest subset of links
that can transmit simultaneously within the model of interference.  We
adopt the SINR model of interference where transmission over a link
succeeds if the received signal at the receiver is sufficiently large,
compared to ambient noise and interference from other transmissions.
This model has emerged as a superior model for wireless interference
patterns, as it is both analytically manageable, and reasonably
realistic, especially in comparison to graph based
models~\cite{GronkMibiHoc01,MaheshwariJD08,Moscibroda2006Protocol}.
We assume that the powers have been pre-assigned to the links, based
only on the length of the links. Having such simple assignments can be of great benefit
in a distributed context.

The basic capacity problem has been addressed in numerous recent
works. Constant-factor approximation algorithms have been given for
uniform power \cite{GHWW09} and more generally for any non-decreasing
sub-linear power assignment (see Sect.~2 for definitions) \cite{SODA11}, and for arbitrary power 
\cite{KesselheimSoda11}. 
These results assume a uniformity of the links, both in their signal
characteristics as well as their value. They also assume that no other
wireless activity is affecting these transmissions. 
We aim to handle more general scenarios, allowing for heterogeneity in
link characteristics and environment.
In particular, we address three extensions:
\begin{enumerate}
 \item (QoS) Each link has its own signal requirements and its own
   ambient noise term.
 \item (Weights) Each link has an associated weight, and the
   objective is to maximize the total weight of the satisfied links.
 \item (Admission control) Certain communication is already taking
   place, which cannot be interfered with (possibly for regulatory reasons).
\end{enumerate}
We discuss each of these extensions further.

\subsection*{Quality-of-Service requirements}

The SINR achieved at a particular link determines the data-rate achieved at this link, or the quality of service (QoS). Different links may have a minimum acceptable QoS requirements, for example if one link is used for video transmission and another for data transmission. In addition, the noise level at receivers may not be the same across the network. This practically motivated version of the capacity problem has not been handled in much of the previous research \cite{GHWW09,HW09,SODA11}.
We tackle this problem, both as an interesting problem in its own right, and as a stepping stone for the following problem.

\subsection*{Cognitive radio and admission control}
Given are two sets of links $\calP$ and $\calL$. The goal is to find  $\calQ \subseteq \calL$ such that $\calP \cup \calQ$ can transmit simultaneously and the size of $\calQ$ is maximized.  We refer to this as the \emph{admission control} problem.

This problem naturally arises in at least two application areas. The first is the so-called ``cognitive radio'', which has been the object of intense study of late (see \cite{Bahl:2009:WSN:1592568.1592573,SuZhangCognitive,Levorato:2009:CIM:1793974.1793991} and their \emph{many} references). This area has gained great salience due to recent 
regulatory changes in wireless bandwidth management. Though the exact technological scenario for cognitive radios is still being figured out, the essential point is as follows: a wireless channel is allocated to a ``primary user''; however, one would like to accommodate more users in the channel, as long as the primary user remains feasible. 
This clearly is an instance of the above mentioned problem, with  $\calP$ being typically small (in fact, perhaps, just $1$).

However, there is a more ``classical'' source of the same problem, referred to as admission control or \emph{access control} \cite{GoldsmithSurvey,WuBertsekas}, sometimes referred to as ``active link protection'' \cite{ChiangSurvey}. The capacity problem, is its basic form, captures a scenario where each slot is independent of previous slots. In practice however, links can require sustained communication (and different links for different periods of time). Thus, in certain applications, a more realistic model is to maximize capacity under the constraint that older links still communicating not be disturbed. This again, is exactly the problem defined above (but perhaps with a typically larger $\calP$). Though heuristic approaches to this problem abound, we are unaware of rigorous algorithmic
results in the SINR model.

\subsection*{Weighted capacity}
In this problem each link $v$ is associated with a non-negative weight $w_v$ and the goal is to find a feasible
set $OPT$ so as to maximize $\sum_{v \in OPT}  w_v$.

This \emph{weighted capacity} problem is a natural extension of the capacity problem, and a case can
be made for theoretical investigation for this reason alone.
As it happens, though, the problem is further motivated by questions about \emph{stability} in queuing theory. 
In this setting, packets arrive at network nodes according to some stochastic process, and the problem is
to characterize the set of arrival rates under which the network can be stabilized, i.e., the network queues
remain bounded. In the case of 
wireless networking stability, the seminal work of Tassiulas and Ephremides \cite{TE92} 
established
the existence policy that
stabilizes the system under all arrival rates for which stability is potentially possible. 
This policy can be seen to be equivalent to solving the maximum
weighted capacity problem in the SINR model. 

\subsection*{Solution method}

It is easy to verify whether a given set of links is feasible. 
In fact, an appropriate power assignment that makes it feasible can be found efficiently. Namely, Eqn. \ref{gen_sinr} can be cast as a linear program with $P_v$'s as variables, 
which can thus be solved optimally. Indeed, there is a large body of work where one starts with a feasible set and then tries to optimize over some other criteria, say to minimize the power consumed \cite{ChiangSurvey}. 

Naturally, one doesn't expect this approach to work directly for the capacity problem introduced before, which is  ``combinatorial'' and in fact happens to be NP-hard \cite{Goussevskaia2008Complexity}.
What is perhaps more surprising is that the capacity problem does not appear to easily admit a linear
programming relaxation either, even for simple cases. Most algorithms developed for the capacity problem have thus been very simple greedy algorithms \cite{GHWW09,SODA11,KesselheimSoda11}, with some exceptions \cite{hoeferspaa,CKMPS08}.

In this work, starting from a simple observation, we develop an integer program that approximates the capacity problem for a large class of oblivious power assignments.
We then show how to round the corresponding linear programming relaxation to get a constant factor approximation.
Thus we recover the main result of \cite{SODA11} but via linear programming as opposed to a greedy algorithm. 
We also show that the LP formulation can be easily modified to tackle a class of important problems where
greedy algorithms do not appear to work very well, including the
problems discussed above.

\section{Preliminaries and results}
\label{sec:model}

The capacity problem in the SINR model is defined as follows.
We are given a set $L$ of $n$ links, 
each consisting of a sender and receiver pair $(s_v, r_v)$, 
which are points in a metric space with a distance metric $d$.
The asymmetric distance from link $w$ to link $v$ is the distance from
$w$'s sender to $v$'s receiver, denoted $d_{wv} = d(s_w, r_v)$.
Each link $v$ has been assigned transmission power $P_v$. 
A link $v$ succeeds if
\begin{equation}
\frac{P_v/d_{vv}^{\alpha}}{N + \sum_{w \in S \setminus \{v\}} P_w/d_{wv}^{\alpha}} \geq \beta\ ,
\label{gen_sinr}
\end{equation}
where $N$ is the ambient \emph{noise}, $\beta$ is the required SINR level,
$\alpha > 0$ is the \emph{path loss} constant, and 
$S \ni v$  is the set of concurrent transmissions. A set $S$ is \emph{feasible} if the above constraint holds for all $v \in S$. Thus the capacity problem is equivalent to finding
the feasible subset $S \subseteq L$ of maximum size.

Let $\ell_v \equiv d_{vv}$ denote the length of link $v$.
Let $\Delta$ denote the ratio between the maximum and minimum length of a link.
A power assignment $P$ is \emph{non-decreasing} if $P_v \ge P_w$ whenever
$\ell_v \ge \ell_w$ and \emph{sub-linear} if $\frac{P_v}{\ell_v^{\alpha}} \le \frac{P_w}{\ell_w^{\alpha}}$ whenever $\ell_v \ge \ell_w$. We will restrict our attention to this class or particular assignments belonging to this class. 
Note that this class essentially contains all ``natural'' length based assignments, and specifically all
 well studied length based power assignments. These include the \emph{uniform} power assignment, where all links use the same power, \emph{linear} power assignment where $P_v = \ell_v^{\alpha}$ (which is thought to be energy efficient), and \emph{mean power} assignment where $P_v = \ell_v^{\alpha/2}$ which is known to be 
essentially the ``best'' length-based assignment as far as capacity is concerned \cite{us:esa09full,SODA11}.

\emph{Affectance. }
We will use the notion of \emph{affectance}, introduced in
\cite{GHWW09,HW09} and refined in \cite{KV10} to the thresholded form
used here, which has a number of
technical advantages.  
The affectance $a^P_w(v)$ \emph{on} link $v$ \emph{from} another link $w$,
with a given power assignment $P$,
is the interference of $w$ on $v$ relative to the power
received, or
\[
a^P_{w}(v)      
     = \min\left\{1, c_v \frac{P_w}{P_v} \cdot \left(\frac{\ell_v}{d_{wv}}\right)^\alpha\right\},
\]
where $c_v = \beta/(1 - \beta N_v \ell_v^\alpha/P_v)$ is a constant
depending only on the parameters of the link $v$. 

We will drop $P$ when clear from context. 
Let $a_v(v) = 0$.
For a set $S$ of links and a link $v$, 
let $a_v(S) = \sum_{w \in S} a_v(w)$ and $a_S(v) = \sum_{w \in S} a_w(v)$.
For sets $S$ and $R$, $a_R(S) = \sum_{v \in R}\sum_{u \in S} a_v(u)$.
Using such notation, Eqn.~\ref{gen_sinr} can be rewritten as follows, which we will adopt:
\begin{equation}
a_S(v) \leq 1
\end{equation}
In the variable QoS version of the capacity problem, $\beta$ and $N$ are no longer constants, but can be different for different links.
Note that the definition of affectance stays the same apart from a changed definition of $c_v = \beta_v/(1 - \beta_v N_v \ell_v^\alpha/P_v)$ where $\beta_v$ and $N_v$ are respectively the signal requirement and noise level for $v$.

For all problems that we consider, we will $OPT$ to mean the optimal solution, which will apply to the problem being discussed at that point.

\subsection*{Our results}
We prove the following results.
\begin{theorem}
For length monotone, sub-linear power assignments, there is constant-approximation algorithm for the wireless capacity problem. For uniform power, there is a constant-approximation algorithm in the QoS generalization.
\label{mainth1}
\end{theorem}
The first part (not involving QoS) is the same as the main result proven in \cite{SODA11}, but via a linear programming relaxation.

\begin{theorem}
For the admission control problem with uniform power,
\begin{itemize}
\item[a)] There is a $O(|\calP|)$ approximation algorithm. 
\item[b)] If the optimum solution $|OPT| > \gamma_1 |\calP| \sqrt{\log  |\calP|}$ (for some constant $\gamma_1$), there is a constant-approximation algorithm.
\end{itemize}
\label{mainth2}
\end{theorem}
Specifically, for the ``cognitive radio'' case of the problem where $\calP = 1$ (or small, at any rate) we get a constant factor
approximation for uniform power. There is no straight-forward greedy algorithm to tackle this problem.
We believe that a greedy algorithm for the variable QoS problem is possible, but even if this is true,
the resultant version for admission control would result in approximation factor worse than our results by a $O(\log n)$ factor. Additionally, we see no way of utilizing the $OPT > |\calP| \sqrt{\log |\calP|}$ condition in the greedy algorithm.

Finally,
\begin{theorem}
For linear power, there is  constant-approximation algorithm for weighted capacity problem.
\label{mainth3}
\end{theorem}
For this problem, greedy algorithms combined with some basic observations can yield  a $O(\min\{\log \Delta, \log n\})$-approximation algorithm (we describe this algorithm in detail in Section \ref{sec:simul} when we experimentally compare it with our LP based algorithm). 

We remark that our results holds in arbitrary metric space, independent of the path loss constant $\alpha$, and faithfully treat the ambient noise term.

\subsection*{Related Work.}

The first work to study capacity of randomly deployed networks was the work of Gupta and Kumar~\cite{kumar00}. 
Rigorous worst case algorithmic analysis started with
the  work of Moscibroda and Wattenhofer \cite{MoWa06}, who
studied of the \emph{scheduling complexity} of arbitrary set
of wireless links. 
Early work on approximation algorithms
 produced approximation
factors that grew with structural properties of the network \cite{moscibroda06b,MoscibrodaOW07,chafekar07}.

The first constant factor approximation algorithm was obtained for
capacity problem for uniform power in \cite{GHWW09} (see also
\cite{HW09}) in $\mathbf{R^2}$ with $\alpha > 2$.
Fangh\"anel, Kesselheim and V\"ocking \cite{FKV09} gave an algorithm
that uses at most $O(OPT + \log^2 n)$ slots for the scheduling problem
with linear power assignment $P_v = \ell_v^\alpha$,
that holds in general distance metrics.

Recently, Kesselheim obtained 
a constant-approximation algorithm
for the capacity problem with power control for doubling metrics
and $O(\log n)$ for general metrics \cite{KesselheimSoda11}. In another work \cite{SODA11}, constant
factor approximation was achieved for all non-decreasing, sub-linear power assignments.
The greedy algorithms of \cite{GHWW09,HW09,SODA11} can be modified to handle the problems
we address here, and these algorithms essentially constitute the previous best results on these 
problems.

As far as we can ascertain, the algorithmic situation for the admission control and weighted capacity
is somewhat similar to the situation the ``basic'' capacity problem was in before the array of results
mentioned above. Thus, we have a large body of works and results in different settings motivating
the study of these questions, but no worst case algorithmic results.

The  works on the emergent field of cognitive radio are too numerous to adequately cover. We refer the reader to  \cite{Bahl:2009:WSN:1592568.1592573} for a  thorough discussion. For results on capacity of networks
in a cognitive radio context see \cite{jafarjournal,shiicccn} etc. (``capacity'' not necessarily meaning the exact same thing we do). For the stability problem in a queuing theory setting
that gives rise to the weighted capacity problem there are many works in graph based models \cite{DBLP:conf/infocom/SharmaMS06,bestInfocom08,DBLP:conf/mobihoc/LiBX09} as well as recent
ones on the SINR model \cite{lqfmobihoc}.  The weighted capacity problem was recently studied in \cite{wanwireless}, where the authors propose a version of the greedy-based $O(\log n)$ approximation. 

In terms of using a LP approach in the SINR setting, there is recent work of Hoefer \emph{et. al.} \cite{hoeferspaa}, using 
related insights in their formulation. In the context of throughput maximization, \cite{CKMPS08} employed a linear programming solution. Being based on unit disc graphs, that approach  does not lead to performance bounds we seek here.

\section{The basic capacity problem}

Let us first consider how one would attempt to write an Integer
program (and a subsequent Linear programming relaxation) for the
capacity problem. If the variable $\delta_v \in \{0, 1\}$ denotes that
link $v$ was selected in the solution, we see that for selected
links, the condition $\sum_{u \in L} a_u(v) \delta_v \leq 1$ would
have to hold. This is quite nice and linear, except for the fact that
we would have to somehow indicate that this condition need only hold
for $\{v : \delta_v = 1\}$, and that no condition need hold for
links in $\{v : \delta_v = 0\}$. There appears no way to do this
in a linear program.

For clarity, we will first present our linear program (and the whole algorithm) below, and then in proving its correctness, we will
describe how out algorithm evades the problem elucidated above.

\subsection{Algorithm}
Our algorithm has three main steps:

\begin{itemize}
\item {\bf Linear Program}
The first step of our algorithm is to solve the following linear program, with variables $\delta_u$, one corresponding to 
every link $\ell_u$.
Let $C$ be a large enough constant.
\begin{align}
\text{(LP)\quad maximize} \sum \delta_u &  \text{ subject to} \nonumber\\
\sum_{v, \ell_v \geq \ell_u} a_v(u) \delta_v & \leq C; \quad \forall u \label{lpcond1}\\
\sum_{v, \ell_v \geq \ell_u} a_u(v) \delta_v & \leq C; \quad \forall u \label{lpcond2}\\
0 \le \delta_u & \le 1; \quad \forall u  \label{lpcond3}
\end{align}

\item {\bf Rounding}
We then ``round" the fractional solution to this linear program in two steps.

Let $LP^*$ be the value of the (fractional) solution to $LP$.

First, we select a set $R = \{u: s_u = 1\}$, defined by binary variables $s_u$, which are generated independently at random such that $\Pro(s_u = 1) = \delta_u$ (and thus $\Pro(s_u = 0) = 1 - \delta_u$). 

Next, we choose a subset of $R$ named $S$ defined by $S = \{u : s'_u = 1\}$,
where $s'_u$ is a binary random variable
corresponding to this second round of selection.
The variable $s'_u$ is defined as follows:
$s'_u = 1$ iff $s_u = 1$ and the following two conditions hold:
\begin{align}
\sum_{v, \ell_v \geq \ell_u} a_v(u) s_v & \leq 3 C \label{rndcond1}\\
\sum_{v, \ell_v \geq \ell_u}a_u(v) s_v  & \leq 3 C \label{rndcond2}
\end{align}

\item {\bf Final Selection}
Finally, a feasible set is extracted from $S$ using a simple \emph{signal-strengthening} technique which we will detail later.
\end{itemize}

\subsection{Analysis}
We need the following definitions.
\begin{defn}
A link set $L$ is \emph{$\gamma$-feasible} (resp., \emph{$\gamma$-anti-feasible}), if $a_L(u) \leq \gamma$ for all $u \in L$ (resp. if $a_{u}(L) \leq \gamma$ for all $u \in L$). A link set is \emph{$\gamma$-bi-feasible} if it is both $\gamma$-feasible and $\gamma$-anti-feasible.
\end{defn}
We will simply write ``feasible'', ``anti-feasible'' and ``bi-feasible" when $\gamma = 1$.

Our first step is to show that the solution to the linear program is an approximation to the 
capacity problem, or more formally:
\begin{lemma}
Let $LP^*$ be the value of the optimal solution of $LP$. Then, $LP^* = \Omega(|OPT|)$.
\label{lpopt}
\end{lemma}
\begin{proof}
To prove this, it suffices to construct a solution $\delta_u$ (for all $u$) to the $LP$ such that
$\sum_u \delta_u = \Omega(|OPT|)$, and that satisfies all the constraints in the linear program.

Since $OPT$ is feasible, there is a 2-bi-feasible subset $W$ such that $|W| > |OPT|/2$ (See \cite{icalp11} for a simple proof of this
fact).

Now construct the solution by setting
$\delta_u = 1$ if $u \in W$ and $\delta_u = 0$ otherwise. Thus $\sum_u \delta_u = |W| \geq |OPT|/2$.  Lemma \ref{lpopt}
can be proven then if we can show that Conditions  \ref{lpcond1} and \ref{lpcond2} hold for this solution, and thus form a valid
solution to $LP$. These follow directly from two Lemmas noted below (Lemmas \ref{cl2} and \ref{cl3}), by setting
$C$ to be larger than the implicit constants in those two Lemmas.
\end{proof}

The following Lemma was proven in \cite{KV10}. For completeness, we give a proof in the
appendix that holds for arbitrary ambient noise.
\begin{lemma}
Assume $L'$ is $\gamma$-feasible using a non-decreasing, sub-linear power assignment. Let $u$ be any link such that $\ell_u \leq \ell_v$ for all $v \in L'$. Then $a_{L'}(u) = O(\gamma)$.
\label{cl2}
\end{lemma}

The next Lemma, something of a dual of the previous one, was proven recently in \cite{icalp11}:
\begin{lemma}
Assume $L'$ is $\gamma$-anti-feasible using a non-decreasing, sub-linear power assignment. Let $u$ be any link such that $\ell_u \leq \ell_v$ for all $v \in L'$. Then $a_u(L') = O(\gamma)$.
\label{cl3}
\end{lemma}

{\bf Remarks:} Lemmas \ref{cl2} and \ref{cl3} hold the crucial insight that allow us to circumvent the problem mentioned at the beginning
of this section. Note how these lemmas bound the affectance to and from a link $u$ without the condition that $\ell_u$ 
be a part of the feasible (or anti-feasible) set $L$. This allows us to evade the issue of having to express conditions
that only apply for links in the solution set. Instead we can write constraints (Equations
\ref{lpcond1} and \ref{lpcond2}) which apply to \emph{all} links.

The next step is to analyze the {\bf Rounding} phase. In particular, we claim that 
\begin{lemma}
$\Ex(|S|) = \Omega(LP^*) = \Omega(|OPT|)$.
\end{lemma}
\begin{proof}
Recall that $S = \{u : s_u' = 1\}$. Then
by linearity of expectation, 

\begin{eqnarray}
\Ex(|S|) &=& \Ex(\sum_u s'_u) = \sum_u \Ex(s'_u) \nonumber \\ &=& \sum_u  \Pro(s'_u = 1) = \sum_u \Pro(s'_u = 1 | s_u = 1) \Pro(s_u = 1) \nonumber \\ &=& \sum_u \Pro(s'_u = 1 | s_u = 1) \delta_u \label{eqExS1} \ .
\end{eqnarray}
where we use $\Pro(s_u = 1) = \delta_u$.

Let $\rho(u)$ denote the indicator random variable of the event that both Cond.\ 
\ref{rndcond1} and \ref{rndcond2} are fulfilled for link $u$. Then  $\Pro(s'_u = 1 | s_u = 1) =  \Pro(\rho(u) | s_u = 1)$.
The point to note here is that the events $\rho(u)$ and $s_u = 1$ are independent since
the random variable $s_u$ is not involved in the former (because $a_u(u) = 0$). 

We will prove below that $\Pro(s'_u = 1) \geq  \frac{1}{3}$ (Lemma \ref{lem:rhobound}).

Thus, $\Pro(s'_u = 1 | s_u = 1)  = \Pro(s'_u = 1) \geq  \frac{1}{3}$. Now continuing with Eqn. \ref{eqExS1},  $\Ex(|S|) \geq \frac{1}{3} \sum_{u} \delta_u\geq \frac{1}{3} LP^*$.

\end{proof}

As promised, we lower bound $\Pro(s'_u = 1)$:
\begin{lemma}
$\Pro(s'_u = 1) \geq \frac{1}{3}$
\label{lem:rhobound}
\end{lemma}
\begin{proof}
By  Eqn.~\ref{lpcond1}, it holds that
$\Ex(\sum_{\ell_v \geq \ell_u} a_v(u) s_v) = \sum_{\ell_v \geq \ell_u}
a_v(u) \Ex(s_v) = \sum_{\ell_v \geq \ell_u} a_v(u) \delta_v \leq C$.
Thus by Markov inequality, the probability that Cond.\ \ref{rndcond1} 
fails is at most $1/3$.
The same applies to Cond.\ \ref{rndcond2}, using Eqn.~\ref{lpcond2}.
The Lemma then follows by the union bound. 
\end{proof}

Finally, we need to show that we can extract a large feasible subset from $S$ in the {\bf Final Selection} phase. 
The following \emph{signal-strengthening} lemma from \cite{HW09}
will be frequently useful.
\begin{lemma}{[Thm. 1 of \cite{HW09}, slightly restated]}
If $S_\gamma$ is an $\gamma$-feasible set, 
then $S_\gamma$ can partitioned in to $O(\left(\frac{\gamma}{\theta}\right)^2)$ $\theta$-feasible sets,
for any $\theta < \gamma$.
\label{lem:signal}
\end{lemma}

\begin{lemma}
There is an efficient algorithm to find a feasible set $S' \subseteq S$ such that $|S'| = \Omega(|S|)$.
\end{lemma}
\begin{proof}
By conditions \ref{rndcond1} and \ref{rndcond2},  the average
affectance in the selected set $S$ is
\begin{align*}
\frac{1}{|S|}\sum_{v, u \in S} a_v(u)
& = \frac{1}{|S|}\sum_{u \in S} \sum_{v \in S, \ell_v \geq \ell_u} 
    \left(a_v(u) + a_u(v)\right) \\
& \leq \frac{1}{|S|}\sum_{u \in S} (3C + 3C)  =  6 C
\end{align*}
Define $S_1 = \{u \in S: a_S(u) \leq 12\}$. From the above bound on average affectance, it is easy to see that $|S_1| \geq |S|/2$.
Finally, by the signal strengthening lemma, we can find a feasible set $S_2$ such that
$|S_2| = \Omega(|S_1|) = \Omega(|S|)$.
 \end{proof}
The first part of Thm.~\ref{mainth1} now clearly follows. The part of Thm.~\ref{mainth1} about uniform power in the variable QoS case will be handled in the next section.

The algorithms in the following two sections will follow the same tri-partite design of LP, Rounding and Final Selection. Due to space
constraints, we will mostly focus on the changes in the LP formulation, and when appropriate, the changes in the Rounding phase, without
proving everything from scratch.

\section{Cognitive radio/Admission control}

\subsection*{Variable noise and signal requirements (QoS)}

Recall that in this variation of the problem each link $v$ has a separate QoS $\beta_v$ and noise level $N_v$ and definition of affectance
changes accordingly.
If a link set is such that $P_v \leq c_1 P_u$ for all $u \ne v$ for some unspecified constant $c_1$, we call the 
power assignment \emph{nearly uniform}. The following holds.
\begin{lemma}
Assume $L'$ is anti-feasible and $u$ is some link. Assume that all
links use a nearly uniform power assignment. Then  $a_{u}(L') = O(1)$.
\label{cl2uni}
\end{lemma}
The proof is a standard modification of the same result for uniform
power with constant $\beta$ and $N$ (see, for example, Lemma 11 of
\cite{infocom11}).  Our proof of Lemma \ref{cl2} provided in the appendix gives a general idea
of this type of proof, and we mention after that proof the main changes needed to achieve Lemma \ref{cl2uni}.

The following modified LP can be used for uniform power capacity in this setting:
\begin{align}
\text{(LP2)\quad maximize} \sum \delta_u &  \text{ subject to} \nonumber\\
\sum_{v \neq u} a_u(v) \delta_v & \leq C; \quad \forall u \\
0 \leq \delta_u & \leq 1; \quad \forall u  \nonumber
\end{align}

The additional steps after solving the LP and the analysis follows the same lines as Thm.~\ref{mainth1} (which we omit due to space constraints).

\subsection*{Admission control}

Now we can focus on the admission control problem for which we will use some ideas
from the variable QoS case.

We will prove the following more general result first.
\begin{theorem}
Assume links in $\calL$ use a nearly uniform power assignment. Assume that links in $\calP$ use some
arbitrary power assignment. Then we can approximate the admission control problem up to a  factor of $O(|\calP|)$.
\label{admingen1}
\end{theorem}

\begin{proof}
Recall that the goal is to find $OPT \subseteq \calL$ of maximum size
such that $OPT \cup \calP$ is feasible. Thus in choosing $OPT$, we
have to be careful about the affectance of $\calP$ on $OPT$ and
vice-versa. Our approach is to handle the affectance from $\calP$ as
noise. In this regime, the new ``noise" present at each link is the
original noise $N$, plus the interference received from all links in
$\calP$. Specifically, for $u \in \calP \cup \calL$,
we define a variable noise level $N_{u} = N + \sum_{v \in
  \calP, u \neq v} \frac{P_v}{d_{vu}^{\alpha}}$.  Define $\hat{a}_u(v)$
to be the affectance taking this variable noise into account.

Now consider the following LP relaxation:
\begin{align}
\text{(LP2)\quad maximize} \sum \delta_u & \text{ subject to} \nonumber\\
\sum_{w \in \calP} \sum_{v \in \calL} \hat{a}_v(w) \delta_v & \leq |\calP|; \quad  \label{cog1} \\
\sum_{v \in \calL} \hat{a}_u(v) \delta_v & \leq C; \quad \forall u \in \calL \label{cog2} \\
0 \leq \delta_u & \leq 1; \quad \forall u \in \calL \nonumber
\end{align}
We show that the solution of LP2 is close to $OPT$.
\begin{lemma}
Let $LP2^*$ be the value of the solution to $LP2$. Then $LP2^* = \Omega(|OPT|)$.
\end{lemma}
\begin{proof}
Consider a $2$-anti-feasible subset of $OPT$, call this $O$.
Consider the following solution to the $LP$: set $\delta_u = 1$ if $u \in O$ and $0$ otherwise.
Cond.\ \ref{cog1} is satisfied since the incoming affectance on 
each link in $\calP$ from $OPT$ (and thus from $O$) is at most $1$. 
The case of Cond.\ \ref{cog2} follows from Lemma \ref{cl2uni}. Thus $LP2^* \geq |O| = \Omega(|OPT|)$.
\end{proof}

The next step is to round the fractional solution achieved from solving the LP. As before, we first set $s_u =1$ with independent probability $\delta_u$.
Let us define the event $\mathcal{A}$ as the condition $\sum_{w \in \calP} \sum_{v \in \calL} \hat{a}_v(w) s_v \leq 5 |\calP|$ holding. Let us define the event $\mathcal{B}_u$, for each link $u \in L$, as the condition $\sum_{v \in \calL} \hat{a}_u(v) s_v \leq 4 C$ holding.

We derive another round of selections by setting  $s_u' = 1$ iff $s_u = 1$ and both $\mathcal{B}_u$ and $\mathcal{A}$ occur. Thus,
 
\[
\Ex(s_u') = \Pro(s'_u = 1) = \Pro(s_u = 1 \land \mathcal{B}_u \land \mathcal{A}) = \Pro(\mathcal{B}_u \land \mathcal{A}|s_u = 1) \delta_u
\]

Now $\Pro(\mathcal{B}_u \land \mathcal{A}|s_u = 1) \geq 1 - \Pro(\bar{ \mathcal{B}}_u |s_u = 1) - \Pro(\bar{\mathcal{A}} |s_u = 1)$. As we have seen before, $\mathcal{B}_u$ is independent of $s_u$, thus $\Pro(\bar{\mathcal{B}}_u |s_u = 1) \leq \frac{1}{4}$ (via  Cond.\ \ref{cog2} and Markov's inequality). 

On the other hand, $\mathcal{A}$ is not independent of $s_u$. However, $\mathcal{A}$ occurring given $s_u = 1$ is the same as  $\sum_{v \in \calL, v \neq u}  \sum_{w \in \calP}\hat{a}_v(w) s_v \leq 5 |\calP| - \sum_{w \in S}\hat{a}_u(w)$ being true. But $\sum_{w \in \calP}\hat{a}_u(w) \leq |\calP|$, by the definition of affectance. Thus $\Pro(\bar{\mathcal{A}} |s_u = 1) \leq \Pro(\sum_{v \in \calL, v \neq u}  \sum_{w \in \calP}\hat{a}_v(w) s_v > 9 |\calP|) \leq \frac{1}{4}$. Thus finally, $\Pro(\mathcal{B}_u \land \mathcal{A}|s_u = 1) \geq 1 - \frac{1}{4} - \frac{1}{4} \geq 0.5$. Therefore, $\Ex(s_u') \geq 0.5 \delta_u$.

After the last round of selection, we thus get a set $R \subseteq \calL$ such that
\begin{itemize}
\item $R = \Omega(LP2^*)$, in expectation
\item $\sum_{w \in \calP} \sum_{v \in R} \hat{a}_v(w) \leq 5 |\calP|$
\item  $\sum_{w \in R} \hat{a}_v(w) \leq 4 C$ for all $v \in R$
\end{itemize}
Using averaging arguments and signal strengthening as before, we can extract $R' = \Omega(R)$ which is feasible.
To complete the solution, we need to extract a subset of $R'$ such that
\begin{equation} 
\sum_{v: s'_v = 1} \hat{a}_v(w) \leq 1 \text{ for all } w \in \calP \label{cogfeas1}
\end{equation} 

From the condition $\sum_{w \in \calP} \sum_{v \in \calL} \hat{a}_v(w) s'_v \leq 10 |\calP|$, it is not hard to see that the set of
selected links from $L$ can be partitioned into $O(|\calP|)$ sets such that Eqn. \ref{cogfeas1} holds. This
gives us the sought after $O(|\calP|)$-approximation.
\end{proof}

Thm.~\ref{admingen1} implies part a) of Thm.~\ref{mainth2} directly.
We note that this implies a $O(|\calP| \log \Delta)$-approximation 
algorithm that holds under any other
non-decreasing sublinear power assignment, by partitioning the
linkset into at most $\log \Delta$ sets of nearly-uniform power.

We prove the last part of Thm.~\ref{mainth2} below.
\begin{theorem}
Let $k = |\calP|$.
If $|OPT| \geq  \gamma_1 k \sqrt{\log k}$, for a large enough constant $\gamma_1$, then there is a constant-factor approximation algorithm for the admission control problem for uniform power.
\end{theorem}

First we show that if $OPT$ is large, we can assume that the affectances from $OPT$ to $\calP$ are small.

\begin{lemma}
Assume $OPT \geq  \gamma_1 k \sqrt{\log k}$, for a large enough constant $\gamma_1$. 
Define $L' = \{u \in \calL: a_u(v) \leq \frac{1}{10 \sqrt{\log k}} \text { for all } v \in \calP\}$ and $OPT' = L' \cap OPT$.
Then $OPT' = \Omega(OPT)$.
\label{lem:saffect}
\end{lemma}
\begin{proof}
To see this, note that $a_{OPT}(\calP) \leq k$, since $\calP$ must be feasible in presence of $OPT$. Now, defining $OL = OPT \setminus OPT'$, $a_{OL}(\calP) \geq |OL| \frac{1}{10 \sqrt{\log k}}$.
Thus $|OL| \frac{1}{10 \sqrt{\log k}} \leq k$, or $|OL| \leq 10 \sqrt{\log k} \cdot k$ and finally $|OPT'| \geq |OPT| - |OL| \geq (\gamma_1 - 10) \sqrt{\log k} \cdot k = \Omega(|OPT|)$ if $\gamma_1$ is large enough.
\end{proof}

We can also claim a strengthening property.
\begin{lemma}
Assume $R$ is a set such that for all $u \in R$, $a_u(v) \leq
\frac{1}{10 \sqrt{\log k}} \text { for all } v \in \calP$, and $a_R(v)
\leq 1$ for all $v \in \calP$. Then there is a subset $\hat R$
with $|\hat{R}| = \Omega(|R|)$,
such that $a_{\hat R}(\ell) \leq 1/3$ for all $\ell \in \calP$ and such a subset can be found in polynomial time with high probability.
\label{lem:sparsify1}
\end{lemma}
\begin{proof}
Simply select each link in $R$ with probability $\frac{1}{6}$. Let the set of selected links be $Q$. Then
$\Ex(a_Q(v)) \leq \frac{1}{6}$ for all $v \in \calP$. Consider a fixed $v \in \calP$. Now
$a_Q(v) = \sum_{u \in R} a_u(v) s_u$ where $s_u$ is the iid random variable indicating selection in to $Q$.
We can use Hoeffding's inequality to get a large deviation bound.

\begin{theorem}{[Hoeffding, \cite{hoeffding1963}]}
Let the independent random variables $X_1, \ldots X_n$ be bounded, i.e., $X_u \in [b_u, d_u]$, and let $X = \sum_u X_u$. Then,
$$\Pro(X - \Ex(X) \geq t) \leq \exp\left(-\frac{2 t^2 n^2}{ \sum_u (d_u - b_u)^2}\right)\ .$$
\end{theorem}

Set $X_u$ to be $a_u(v) s_u$. We can verify that given our assumptions, setting $b_u = 0$ and $d_u = \frac{1}{10 \sqrt{\log k}}$ suffices. Setting $t = \frac{1}{6}$,
\begin{align*}
 \Pro\left(a_Q(v) \geq \frac{1}{3}\right) & \leq \Pro\left(a_Q(v) - \Ex(a_Q(v)) \geq \frac{1}{6}\right) \\ 
&\leq  \exp\left(-\frac{2  n^2}{36 n^2 (\frac{1}{10 \sqrt{\log k}})^2}\right) \leq \frac{1}{10 k}
\end{align*}
This implies, by the union bound, that with probability at least $\frac{9}{10}$, $a_Q(v) \leq \frac{1}{3}$ for all $v \in \calP$ simultaneously. We now have proof of not only the existential statement, but the algorithmic one, since we can repeat the random experiments multiple times to get the high-probability result. 
\end{proof}
Note that the above holds equally for any affectance function (specifically, the case $\hat a$).

Now we can describe the linear programming relaxation. First note that by virtue of Lemma \ref{lem:saffect} it suffices to  assume the input instance is $L'$ and the optimum is $OPT'$ (links not in $L'$ can be thrown out by simple pre-processing). 
Let us reuse notation $\calL$ and $OPT$ to refer to this new instance after pre-processing.

Consider the following linear program ($LP4$)
\begin{align}
\text{(LP4)\quad maximize} \sum \delta_u & \text{ subject to} \nonumber\\
\sum_{v \in \calL} \hat{a}_v(w) \delta_v & \leq 1/3; \quad \forall w \in \calP \label{cog1l} \\
\sum_{v \in \calL} \hat{a}_u(v) \delta_v & \leq C; \quad \forall u \in \calL \label{cog2l} \\
0 \leq \delta_u & \leq 1; \quad \forall u \in \calL \nonumber
\end{align}
We claim that this is a relaxation up to constant factors.
\begin{lemma}
Let $LP^*$ be the optimal value of $LP4$. The $LP^* = \Omega(|OPT|)$.
\end{lemma}
\begin{proof}{(Sketch)}
By Lemma \ref{lem:sparsify1}, there exists  $O \subseteq OPT$ such that $a_O(w) \leq \frac{1}{3}$ for all $w \in \calP$. Now consider
a 2-anti-feasible subset $O'$ of $O$. Consider as the solution of the
$LP$, $\delta_u = 1$ if $u \in O'$ and $\delta_u = 0$
otherwise. That Cond.\ \ref{cog1l} is satisfied follows from the claim
that $a_O(w) \leq 
\frac{1}{3}$ for all $w \in \calP$. Cond.\ \ref{cog2l} follows from anti-feasibility arguments along the lines made before.
\end{proof}

We can round this solution in the same way as before, with a two stage selection process. The proof varies only in that we need to claim that after the second
selection (characterized by Bernoulli variable $s_u'$), $\sum_{v \in \calL} \hat{a}_v(w) s'_v \leq 1$ with high probability, simultaneously for all $w \in \calP$.
This  follows from an argument similar to Lemma \ref{lem:sparsify1} using the fact that affectances are bounded by $\frac{1}{10 \sqrt{\log k}}$ and using the Hoeffding's inequality.

\section{Weighted capacity}
\label{sec:weighted}
Recall that the result for weighted capacity applies only to linear power. For linear power the following stronger version of Lemma \ref{cl2} holds.

\begin{lemma}
Assume $L'$ is feasible using linear power and $\ell$ is any link (also using linear power). Then, $a_{L'}(\ell) = O(1)$.
\label{cl4}
\end{lemma}
The proof is nearly identical to that of Lemma \ref{cl2}, as elaborated
in the appendix.

We can now write the following LP relaxation for the weighted capacity problem.
\begin{align}
 \text{(LP5) maximize} \sum w_u \delta_u &  \text{ subject to} \nonumber\\
  \sum_{v} a_v(u) \delta_v & \leq C; \quad \forall u \label{inaffectancebound1}\\
  0 \leq \delta_u & \leq 1; \quad \forall u \in L \nonumber
\end{align}

\begin{proof}{of Thm.~\ref{mainth3} (sketch)}
The proof is rather like that of Thm.~\ref{mainth1}. Once again, we select each link into a set $R$ with probability $\delta_u$ (characterized by Bernoulli variable $s_u$ for each link $u$) and then do a further selection by setting $S = \{u : s'_u = 1\}$ where $s'_u = 1$ iff $s_u = 1$ and $\sum_{v \in R} a_v(u) \leq 4 C$. As in the proof of Thm.~\ref{mainth1}, one can show that $\Ex(s'_u) = \Omega(\delta_u)$ and thus the expected weighted output is $\Omega(\sum_{u} w_u \delta_u)$ which is within a constant factor of the optimum. Finally, the set $S$ can be partitioned into a constant number of
feasible subsets using signal strengthening (Lemma \ref{lem:signal}), completing the proof.
\end{proof}

For other power assignments, such as uniform power, we seem to be within striking distance of a $O(1)$-approximation.
This is unfortunately not the case, we can only claim a poly-logarithmic approximation, worse than the greedy case. However, as
we show in Section \ref{sec:simul}, in practice the LP approach might be applicable to these other power assignments as well.

%
%
%
  
\section{Simulations}
\label{sec:simul}
In this section, we present results from simulation experiments. We focus on the weighted capacity problem for our experiments. It is difficult to conduct a comparative experiment for the admission control problem, there being no obvious previous algorithm to compare it with.

In contrast, the weighted capacity problem admits straightforward modifications of to the greedy algorithm, and thus a better comparative benchmark for our algorithm. Two natural greedy algorithms can be proposed:

\begin{itemize}
\item {\bf Using weight classes}: Let $\max_u w_u = n$ (by scaling). Now we can assume that all $w_u \in [1, n]$. This is because links with smaller weights can be discarded without losing more than a factor of $2$ in the approximation quality. Now divide the links into
$\log n$ weight classes, the weight class $W_t$ is defined by $W_t = \{u : w_u \in [2^t, 2^{t+1}]\}$ for $t = 0$ to $\log n -1$. Now if we consider links belonging to a single $W_t$, the weights do not matter (up to a factor of 2). We simply run the greedy algorithm of \cite{SODA11} for each $W_t$, and output the solution for the best weight class. This gives a straightforward $O(\log n)$ approximation factor.
\item {\bf Using length classes}: Let $\min_u \ell_u = 1$ by scaling and let $\max_u \ell_u = \Delta$. Divide the links
into length classes $L_t = \{u : \ell_u \in [2^t, 2^{t+1}]\}$ for $t = 0$ to $\log \Delta$. Within $L_t$ we can choose to run the greedy algorithm on the links in any order since the lengths are essentially the same, thus we go through links according to descending order of weights, achieving a constant factor approximation on $L_t$. We choose the solution for the best $L_t$, thus getting a $O(\log \Delta)$ approximation.
\end{itemize}
Thus, comparing the two greedy algorithms, we achieve a $O(\min(\log
n, \log \Delta))$ approximation. In what follows, we shall refer to
this joint algorithm as 
``greedy algorithm''.

\subsection*{Experimental setup}
We randomly generated the instances. Some important parameters of the experiments are as follows:
\begin{enumerate}
\item $\Delta$: The maximum length of a link (the implicit minimum being $1$)
\item $R$: A number  indicating that the sender of a link is chosen
  from a $R \times R$ square
\item $n$: number of links
\end{enumerate}
We also use $N=0$, $\beta=1$, and $\alpha = 2.5$.
The instances were generated as follows. For each link, the sender was chosen randomly from a $R \times R$ square. The length of the link was chosen randomly from $[1, \Delta]$. The receiver was thus placed at this distance from the sender and at a random direction. The weight was chosen independently from $[1, n]$. We will mention different weight distributions later, and mention this initial choice of weight distribution as the \emph{ordinary} distribution.

One crucial aspect of both greedy algorithms as well as the LP algorithm is the constants used. For the LP algorithm, this is the constant $C$ in Eqn. \ref{inaffectancebound1}. The greedy algorithm of \cite{SODA11} also depends on a constant.
Though theoretical bounds for these constants are available, it has been observed before that these theoretical
bounds do not perform the best in practice \cite{infocom11}. We run all algorithms with different values of the constant in question, running over reasonable values in small increments, and choosing the best solution for each algorithm separately.
We ran our experiments in MATLAB, and used the convex optimization package  \texttt{CVX} \cite{cvx} to solve the LP.

The overall message from the experiments is that using the linear
programming formulation gives a substantial improvement in the
solution quality in many cases. 
On the other hand, the greedy algorithm is also not without merit, and can outperform the LP in certain other situations.
It appears that the smaller the maximum feasible set is, the better
greedy does, while as the solution size/quality improves, LP
outperforms greedy. This is not surprising. When the set is really
dense and the link lengths are large, the quality of the solution is
bad and the cost incurred by greedy due to length-class or
weight-class partitions is minimal. 

 \begin{figure}
\begin{center}
\includegraphics[height=2in]{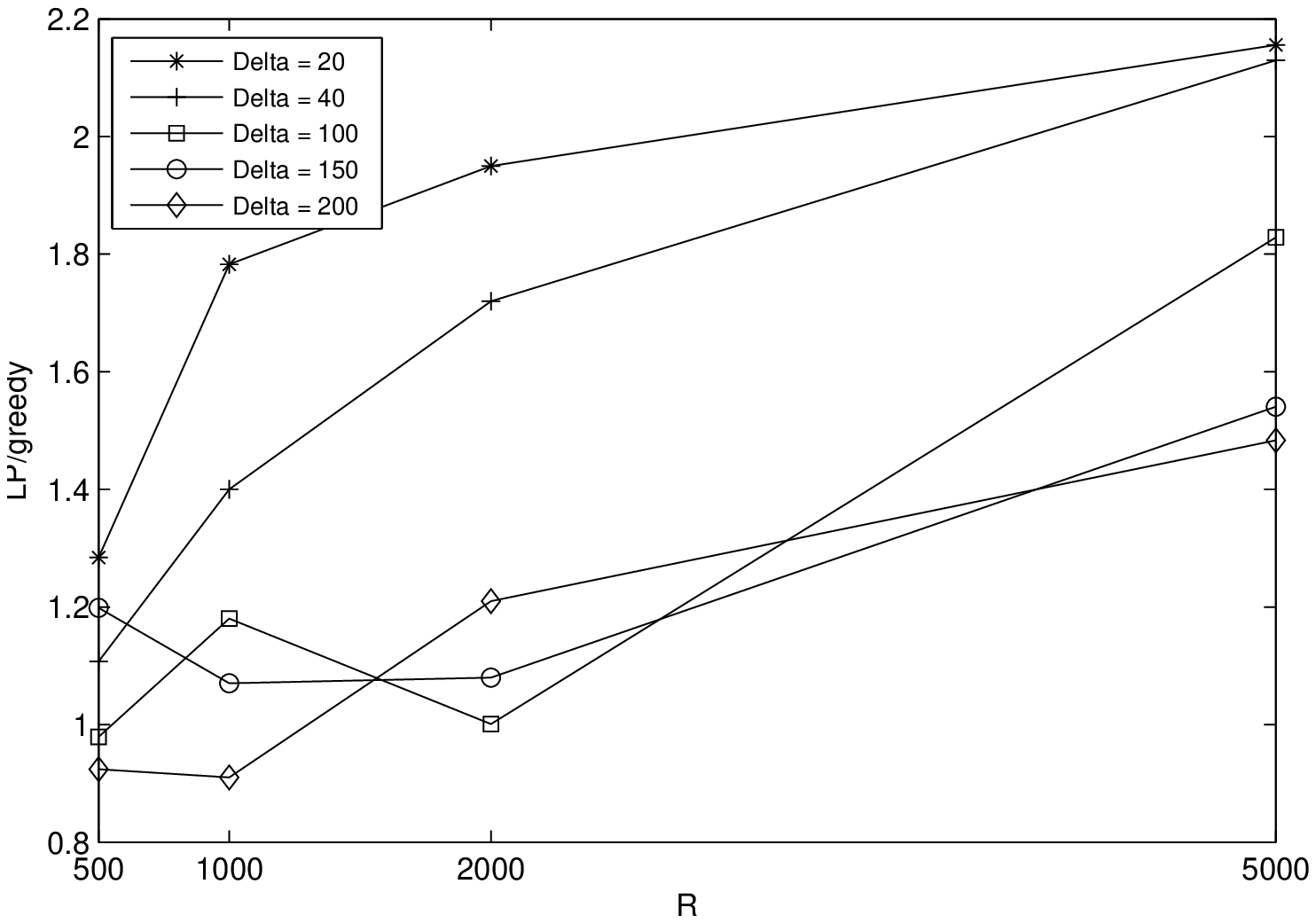}
\caption{Simulation with linear power and $n = 400$ (ordinary weight assignment). Individual lines refer to different values of $\Delta$. The ratio of the solution from the LP algorithm to the greedy algorithm is plotted on the Y-axis against the density.} 
\label{fig:linstraight400}
\end{center}
\end{figure}

\begin{figure}
\begin{center}
\includegraphics[height=2in]{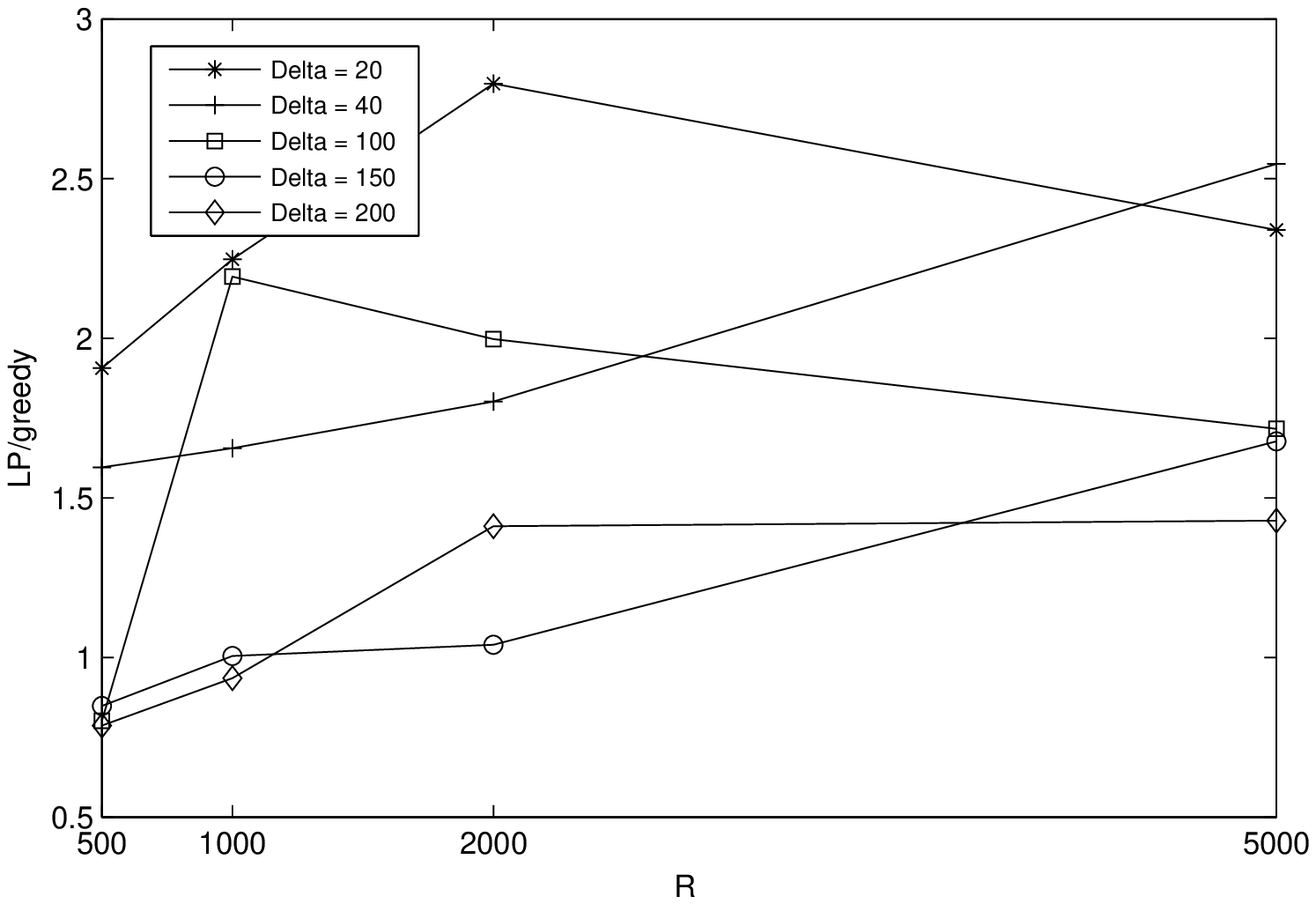}
\caption{Simulation with linear power and $n = 600$ (ordinary weight assignment). } 
\label{fig:linstraight600}
\end{center}
\end{figure}

\begin{figure}
\begin{center}
\includegraphics[height=2in]{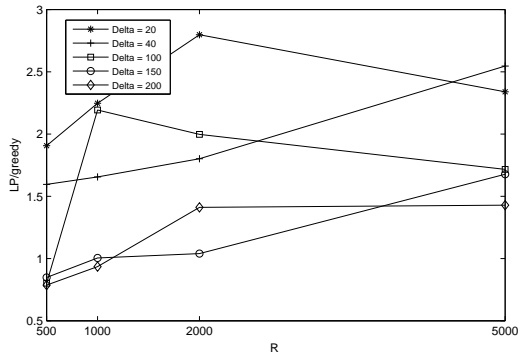}
\caption{Simulation with linear power and $n = 600$ and {\bf Reversed} weight assignment. } 
\label{fig:linreverse600}
\end{center}
\end{figure}

\begin{figure}
\begin{center}
\includegraphics[height=2in]{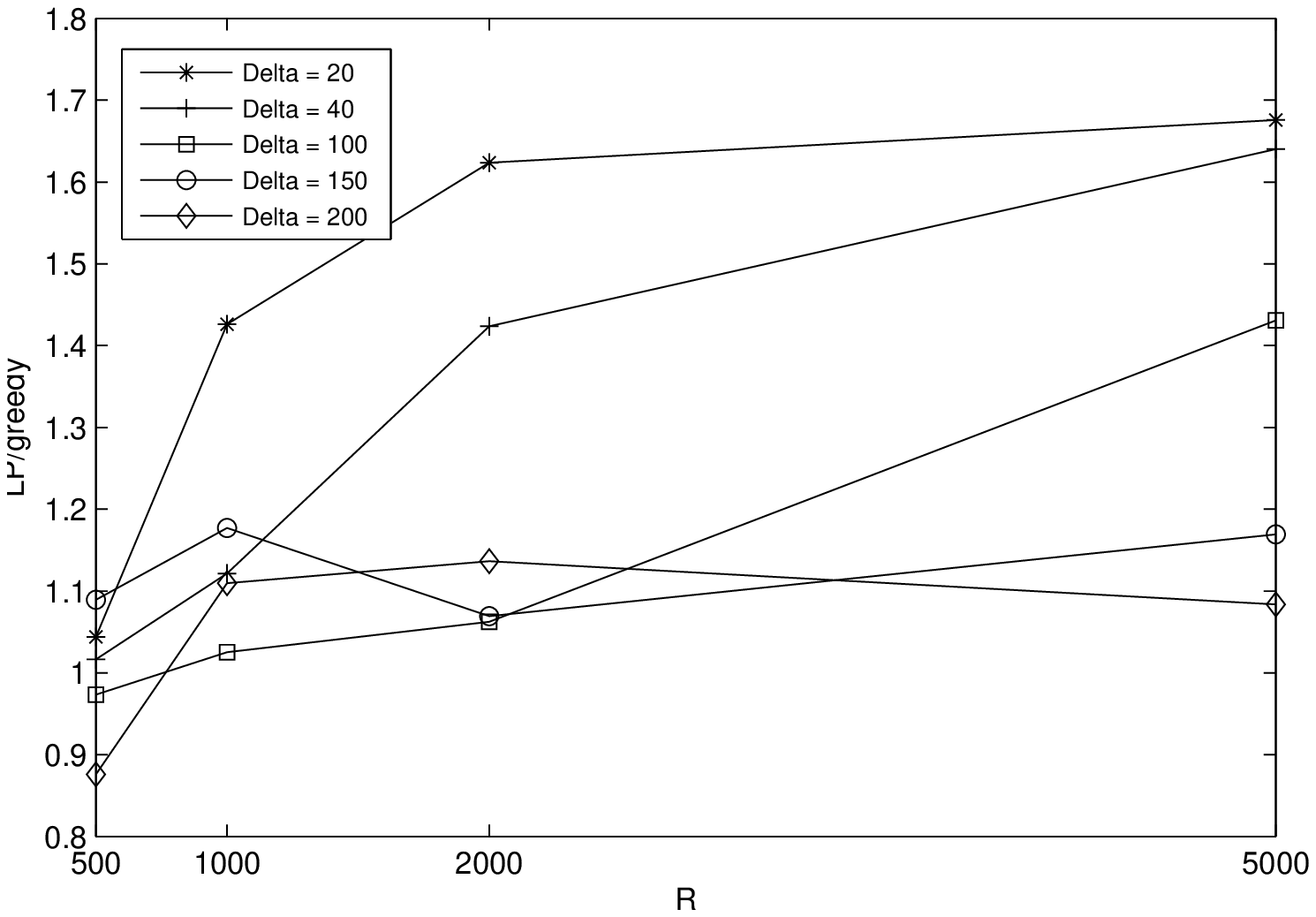}
\caption{Simulation with uniform power and $n = 400$ and the {\bf ordinary} weight assignment. } 
\label{fig:unistraight400}
\end{center}
\end{figure}

In Fig. \ref{fig:linstraight400}, we see the results for linear power with $400$ links. On Y-axis is plotted $\frac{SLP}{SG}$, where
$SLP$ and $SG$ are, respectively, the quality of the solution found from the linear programming algorithm and the greedy algorithm.
As alluded before, the greedy algorithm does better when $\Delta$ and density are both large (these are the points for which the Y-axis value is lesser than $1$), with the trend reversing when these change. 
Running the same experiment run for an increased number of links $n = 600$ confirms these trends (Fig. \ref{fig:linstraight600}).

We experimented with different distributions on the weights, to see if changes here change the solution trend significantly. We tried the following weight distributions.

\begin{itemize}
\item {\bf Reversed}: Set the weight of link $v$ to be $1/w_v$ where $w_v$ is chosen according to the ordinary distribution.
\item {\bf Length determined}: Set weight of the link to be equal to its length.
\item {\bf Weight class}: Choose a parameter $t$ randomly from $[1, \log n]$ and set weight to $2^t$.
\end{itemize}
The overall trend is similar. For reversed and  length determined distributions, LP did extremely well, whereas for the case
of weight class distribution, greedy did much better, with LP only barely outperforming it in a few cases. This further points to
the benefit of combining these algorithms in practice. The results for the reversed case are shown in Fig. \ref{fig:linreverse600}.

Next we experimented with uniform power. As we discussed in Section \ref{sec:weighted}, for uniform power we can only claim a poly-logarithmic approximation factor. However, the bounds are only so bad on rather pathological instances and one needs to do some work to come up with them. Thus in practice, it is reasonable to assume that an LP approach will be not without benefit.
  This is indeed borne out by our experiments, as seen in Fig. \ref{fig:unistraight400}.

\subsection*{Acknowledgements}
Research partially funded by grant 90032021 and grant-of-excellence 120032011 from the Icelandic Research Fund.
Authors thank Neal Young for helpful discussions.

\vspace*{20pt}

\bibliographystyle{plain}
\bibliography{references}

\appendix

We give a proof of Lemma \ref{cl2}, originally due \cite{KV10}, that holds
also in the presence of arbitrary noise.

\emph{Lemma \ref{cl2}:}
If $L$ is $\gamma$-feasible using a non-decreasing sublinear power
assignment and  $u$ is a link such that $\ell_u \leq \ell_v$ for all $v \in L$, then $a_{L}(u) = O(\gamma)$.
\smallskip

\begin{proof}
Assume that $L$ is a $1/3^\alpha$-feasible set. By the signal
strengthening (Lemma \ref{lem:signal}), this affects only the constant factor.

Consider the link $v \in L$ such that $d(r_u, r_v)$ is minimum.
Also consider the link $w \in L$ with $d(s_{w}, r_u)$ minimum.
Let $D = d(r_u, r_v)$. We claim that for all links $x$ in $L$, $x \ne w$,
\begin{equation}
d(s_x, r_u) \geq \frac{1}{2} D \ .
\label{eqn:dist2}
\end{equation}
To prove this, assume, for contradiction, that $d(s_x, r_u) < \frac{1}{2} D$. 
Then,  $d(s_w, r_u)  < \frac{1}{2} D$, by definition of $w$. 
Now, again by the definition of $v$, $d(r_{x}, r_u) \geq D$ 
and $d(r_w, r_u) \geq D$. 
Thus $\ell_w \geq d(r_u, r_w) - d(s_w, r_u) > \frac{D}{2}$
and similarly $\ell_{x} > \frac{D}{2}$. 
On the other hand $d(s_{w}, s_{x}) \le d(s_w, r_u) + d(s_x, r_u) < 
   \frac{D}{2} + \frac{D}{2} < D$.
Now, $d_{w x} \cdot d_{x w} 
  \leq (\ell_{x} + d(s_{w}, s_{x}))(\ell_w + d(s_w, s_{x})) 
  < (\ell_x + D)(\ell_w + D)
  < 9 \ell_{w} \ell_x$, contradicting the following:
  
  \begin{lemma}[\cite{us:esa09full}]
Let $u, v$ be links in a $1/q^\alpha$-feasible set.
Then, $d_{uv} \cdot d_{vu} \ge q^2 \cdot \ell_u \ell_v$. 
\label{lem:ind-separation}
\end{lemma}

Consider now any link $x$ in $L$, $x \ne w$.
By the triangle inquality and Eqn.~\ref{eqn:dist2}, 
$d_{x v}  = d(s_x, r_v) \leq d(r_v, r_u) + d(s_x, r_u) = D + d(s_x, r_u)
   \leq 3 d(s_x, r_u) = 3 d_{x u}$. 
Now $a_{x}(u) \leq c_u \frac{P_{x}}{d_{x u}^{\alpha}} \frac{\ell_u^{\alpha}}{P_u}$. 
Since $\ell_u \leq \ell_v$, it holds that 
$c_u \le c_v$ and by sub-linearity it holds that
$P_u/\ell_u^\alpha \ge P_v/\ell_v^\alpha$.
Thus, 
\begin{align}
a_{x}(u) \leq c_v \frac{P_x}{d_{xu}^{\alpha}} \frac{\ell_v^{\alpha}}{P_v} 
   \leq c_v \frac{3^{\alpha} P_x}{d_{xv}^{\alpha}} \frac{\ell_v^{\alpha}}{P_v} 
  = 3^{\alpha} a_v(x),
\label{eq:affs1}
\end{align}
where the final equality follows from the feasibility of $L$.
Finally, summing over all links in $L$
\begin{align*}
a_L(u) & = \sum_{x \in L}  a_{x}(u)  = a_{w}(u) + \sum_{x \in L \setminus \{w\}}  a_{x}(u)\\
 & \leq 1 + 3^\alpha \sum_{x \in L \setminus \{w\}}  a_u(v) 
  \leq 1 + 3^\alpha \cdot \gamma = O(1) \ , 
\end{align*}
since $\sum_{x \in L \setminus \{w\}}  a_{x}(u) \leq a_L(u) \leq \gamma$ by assumption.
\end{proof}

\medskip

\emph{Proofs of other lemmas:} When using linear power, it holds for all links $u$ and $v$ that
$c_u = c_v$ and the signal received $P_u/\ell_u^\alpha = P_v/\ell_v^\alpha$,
satisfying Eqn.~\ref{eq:affs1} holds without the condition $\ell_u \le \ell_v$.
This yields Lemma \ref{cl4}.

By inverting the role of senders and receivers, we can obtain similar
bounds on out-affectance ($a_u(L)$) as above on in-affectance ($a_L(u)$).
Modulo this change, the proof of Lemma \ref{cl3} is nearly identical
to the above proof of Lemma \ref{cl2}, and the argument for Lemma
\ref{cl2uni} mirrors that of Lemma \ref{cl4}.

\end{document}